\documentclass{article}
\usepackage{spconf,amsmath,graphicx,times,amssymb,dsfont,float}
\usepackage{bm,rotating,lscape}
\usepackage{latexsym,ifthen,bm}
\usepackage{amsthm}
\usepackage{amsfonts}
\usepackage{graphicx}
\usepackage{color}
\usepackage{slashbox}
\usepackage{pict2e}

\newcommand{\lsN}{{\log\sqrt{N}}}
\newcommand{\ep}{\varepsilon}
\newtheorem{Theorem}{Theorem}
\title{Error Probability Bounds for Binary Relay Trees with Crummy Sensors}

\name{Zhenliang Zhang$^{\star}$  Ali Pezeshki$^{\star}$  William~Moran$^{\dagger}$  Stephen~D.~Howard$^{\ddag}$ and Edwin~K.~P.~Chong$^{\star}$ \thanks{This work was supported in part by AFOSR under Contract FA9550-09-1-0518, and by NSF under Grants ECCS-0700559, CCF-0916314, and CCF-1018472.}}

\address{$^{\star}$ Dept. of Elec. and Comp. Engineering, Colorado State University, Fort Collins, CO 80523, USA \\
    $^{\dagger}$ Dept. of Elec. Engineering, The University of Melbourne, Melbourne, Vic. 3010, AU \\
    $^{\ddag}$ Defence Science and Technology Organization, P.O. Box 1500, Edinburgh 5111, AU}

\begin{document}

\maketitle

\begin{abstract}
We study the detection error probability associated with balanced binary relay trees, in which sensor nodes fail with some probability. We consider $N$ identical and independent crummy sensors,
represented by leaf nodes of the tree.  The root of the tree
represents the fusion center, which makes the final decision between two
hypotheses. Every other node is a relay node, which fuses at most two binary
messages into one binary message and forwards the new message to its
parent node. We derive tight upper and lower bounds for the total error probability
at the fusion center as functions of $N$ and characterize how fast the total
error probability converges to 0 with respect to $N$. We show that the
convergence of the total error probability is sub-linear, with the same decay exponent as
that in a balanced binary relay tree without sensor failures. We also show that the total
error probability converges to 0, even if the individual sensors have total error probabilities that converge to $1/2$ and the failure probabilities that converge to $1$, provided that the convergence rates are sufficiently slow.
\end{abstract}


\section{Introduction}

Consider the \emph{decentralized detection} problem introduced in~\cite{Tenney}: Each sensor makes a measurement and summarizes its measurement into a message. These messages are forwarded to the fusion center, which then makes a final decision.

This decentralized  detection problem has been studied in the context of several different network topologies.
In the \emph{parallel architecture}, also known
as the \emph{star architecture} \cite{Tenney}--\nocite{Chair,Cham,dec,Tsi,Tsi1,Warren,Vis,Poor,Sah,chen1,Liu,chen,Kas}\cite{Chong},\cite{BOOK}, all sensors directly communicate with the fusion center.
When sensor measurements are conditionally independent, the decay rate of the total error probability in the parallel architecture is exponential~\cite{Tsi1}.

Another well-studied configuration is the tandem network  \cite{Tang}--\nocite{Tum,tandem,athans}\cite{Venu},\cite{BOOK}. The decay rate of the error probability in this case is sub-exponential \cite{Venu}. Furthermore, as the number of sensors $N$ goes large, the error probability is $\Omega(e^{-cN^d})$ for some positive constant $c$ and for all $d\in (1/2, 1)$~\cite{tandem}. This configuration represents a situation where the length of the network is the longest possible among all networks with $N$ nodes.

The configuration of bounded-height tree has been studied in \cite{Tang1}--\nocite{Nolte,tree1,tree2,tree3,Pete,Alh,Gubner,Will}\cite{Lin},\cite{BOOK}.
This configuration reduces the transmission cost compared to the parallel configuration.
In the bounded-height tree structure, leaf sensor nodes summarize their measurements and send the new messages to their parent nodes, each of which
fuses all the messages it receives with its own measurement (if any) and then forwards the new message to its parent node at the next level.
This process takes place throughout the tree culminating in the fusion center, where a final decision is made. If only the leaf nodes are sensors making measurements, and all other nodes simply fuse the messages received and forward the new messages to their parents, this tree is known as a \emph{relay tree}. For a bounded-height tree with $\lim _{\tau_N\rightarrow \infty} \ell_N/\tau_N=1$, where $\tau_N$ denotes the total number of nodes and
$\ell_N$  denotes the number of leaf nodes, the optimum error exponent is the same as that of the parallel configuration~\cite{tree1}.

For trees with unbounded height, the convergence
analysis is still largely unexplored. In \cite{Gubner}, the convergence of the total error probability in balanced binary relay trees with unbounded height has been proved. Upper and lower bounds for the total error probability at the fusion center as functions of $N$ have been derived in \cite{Zhang}.
These bounds reveal that the convergence of the total error probability at the fusion center is sub-linear with a decay exponent $\sqrt{N}$.

In this paper, we assume that each of the sensors fails with a certain probability. A failed sensor will not provide a message to its parent node at the next level. We refer to these sensors as \emph{crummy}\footnote{The attentive reader will recognize that our use of the term \textquotedblleft crummy\textquotedblright follows in the footsteps of our great patriarch, Claude E. Shannon \cite{Shannon}.} sensors.  We will derive upper and lower bounds for the total error probability at the fusion center as functions of $N$. Not surprisingly, we find that the decay of the total error probability for each step is worse than the case where there is no sensor failure. But this decay rate is still sub-linear with the same decay exponent $\sqrt{N}$ in the asymptotic regime, regardless of the sensor failure probability.

\section{Problem Formulation}
We consider the problem of binary hypothesis testing between $H_0$ and $H_1$ in a balanced binary relay tree with crummy sensors. As shown in Fig. \ref{fig:tree}, leaf nodes are sensors undertaking initial and independent detections of the same event in a scene. These measurements are summarized into binary messages. If a sensor node works properly, then it forwards the summarized message to its parent node at the next level. Otherwise, with a certain probability the sensor fails in the sense that it does not forward the message upward. Each non-leaf node---except the root, which is the fusion center---is a relay node, which fuses binary messages it receives (if any, and at most two) into one new binary message and forwards the new binary message to its parent node. This process takes place at each intermediate node culminating the fusion center, at which the final decision is made based on the information received.

\begin{figure}[htbp]
\centering
\includegraphics[width=3.5in]{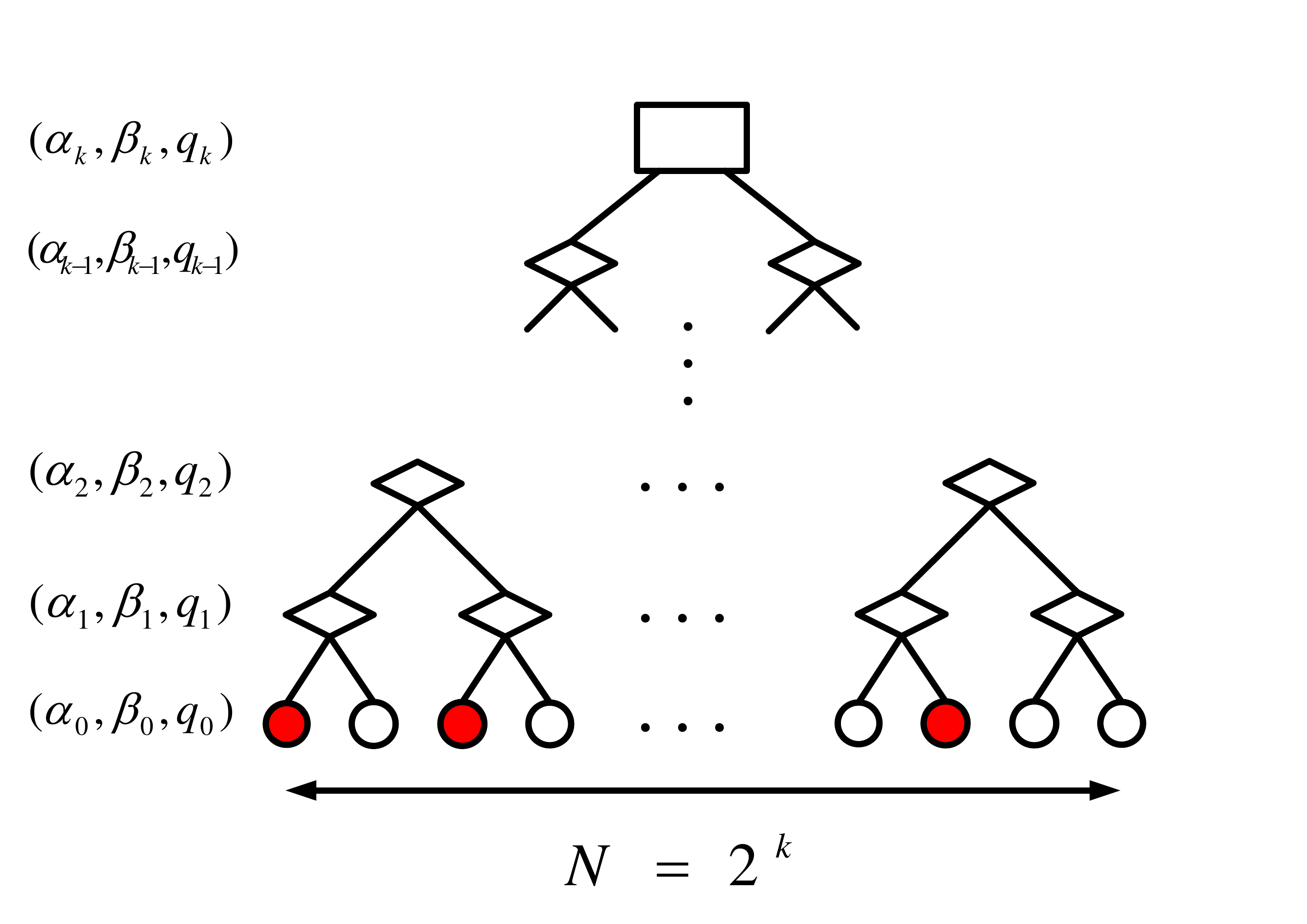}
\caption{A balanced binary relay tree with height $k$. Circles represent sensors making measurements. Diamonds represent relay nodes which fuse binary messages. The rectangle at the root represents the fusion center making an overall decision.}
\label{fig:tree}
\end{figure}
We assume that all sensors are independent given each hypothesis, and that all sensors have identical Type I error probability $\alpha_0$ and identical Type II error probability $\beta_0$. Moreover, we assume that all sensors have identical failure probability $q_0$. Assuming equal prior probabilities, we use the likelihood-ratio test \cite{VT} when fusing binary messages at intermediate relay nodes and the fusion center.

Consider the simple problem of fusing binary messages passed to a node by its two immediate child nodes. Assume that the two child nodes have identical Type I error probability $\alpha$, identical Type II error probability $\beta$, and identical failure probability $q$.

Denote the Type I error, Type II error, and failure probabilities after the fusion by $(\alpha', \beta',q')$.
This parent node fails to provide any message to the node at the next level if and only if both its child nodes fail to forward any message. Hence, we have

\begin{equation}
q'=q^2.
\label{equ:q}
\end{equation}

If one of the child nodes fails and the other one sends its message to the parent node, then Type I and Type II error probabilities do not change since the parent node receives only one binary message. The probability of this event is $2q(1-q)$, in which case we have
\begin{equation}
(\alpha', \beta')=(\alpha,\beta).
\end{equation}

If both child nodes send their messages to the parent node, then the scenario is the same as that in \cite{Gubner} and \cite{Zhang}. The probability of this event is $(1-q)^2$, in which case we have
\begin{equation}
(\alpha', \beta')=\left\{\begin{array}{c}
(1-(1-\alpha)^2, \beta^2),\quad \alpha\leq\beta,\\
 \\
(\alpha^2, 1-(1-\beta)^2),\quad \alpha>\beta.\end{array}\right. \\
\label{equ:rec}
\end{equation}

Let $\bar{\alpha'}$ and $\bar{\beta'}$ be the mean Type I and Type II error probabilities conditioned on the event that at least one of these child nodes forwards its message to the parent node, i.e., the parent node has data. We have
\begin{eqnarray}
&&(\bar{\alpha'}, \bar{\beta'},q')=f(\alpha,\beta,q)\\
&&=\left\{\begin{array}{c}
\left(\frac{(1-q)(2\alpha-\alpha^2)+2q\alpha}{1+q}, \frac{(1-q)\beta^2+2q\beta}{1+q}, q^2\right),  \alpha\leq\beta,\\
 \\
\left(\frac{(1-q)\alpha^2+2q\alpha}{1+q}, \frac{(1-q)(2\beta-\beta^2)+2q\beta}{1+q}, q^2\right),  \alpha>\beta.\end{array}\right.
\label{equ:exp}
\end{eqnarray}

Our assumption is that all sensors have the same error probabilities $(\alpha_0,\beta_0,q_0)$. Therefore by (\ref{equ:exp}), all relay nodes at level $1$ will have the same \emph{error probability triplet} $(\alpha_1,\beta_1,q_1)=f(\alpha_0,\beta_0,q_0)$ (where $\alpha_1$ and $\beta_1$ are the conditional mean error probabilities). Similarly by (4), we can calculate error probability triplets for nodes at all other levels. We have
\begin{equation}
(\alpha_{k+1}, \beta_{k+1}, q_{k+1})=f(\alpha_k, \beta_k, q_k), \quad k=1,2,\ldots,
\label{equ:rel}
\end{equation}
where $(\alpha_k,\beta_k,q_k)$ is the error probability triplet of nodes at the $k$th level of the tree. Notice that if we let $q_0=0$, then the recursive relation reduces to the recursion in \cite{Zhang}.

The relation (\ref{equ:rel}) allows us to consider ${(\alpha_k,\beta_k,q_k)}$ as a discrete dynamic system. For the case where $q_0=0$, we have studied (See \cite{Zhang}) the precise evolution of the sequence $\{(\alpha_k, \beta_k)\}$, derived total error probability bounds as functions of $N$, and established asymptotic decay rates. In this paper, we will study the case where $q_0\neq 0$. We will derive total error probability bounds and determine the decay rate of the total error probability.

To develop intuition, let us start by looking at the single trajectory shown in Fig. \ref{fig:plane}(a), starting at the initial state $(\alpha_0, \beta_0,q_0)$.
We observe that $q_k$ decreases very fast to 0. In addition, as shown in Fig. \ref{fig:plane}(b), the trajectory approaches $\beta=\alpha$ at the beginning. After $(\alpha_k,\beta_k)$ gets too close to $\beta=\alpha$, the next pair $(\alpha_{k+1},\beta_{k+1})$ will be repelled toward the other side of the line $\beta=\alpha$. This behavior is similar to the scenario where $q=0$. For the case where $q=0$, there exist an invariant region in the sense that the system stays in the invariant region once the system enters it \cite{Zhang}. Is there an invariant region for the case where $q\neq0$? We answer this question by precisely describing this invariant region in $\mathds{R}^3$.
\begin{figure}[!th]
\begin{center}
\begin{tabular}{cc}
\includegraphics[width=2.2in]{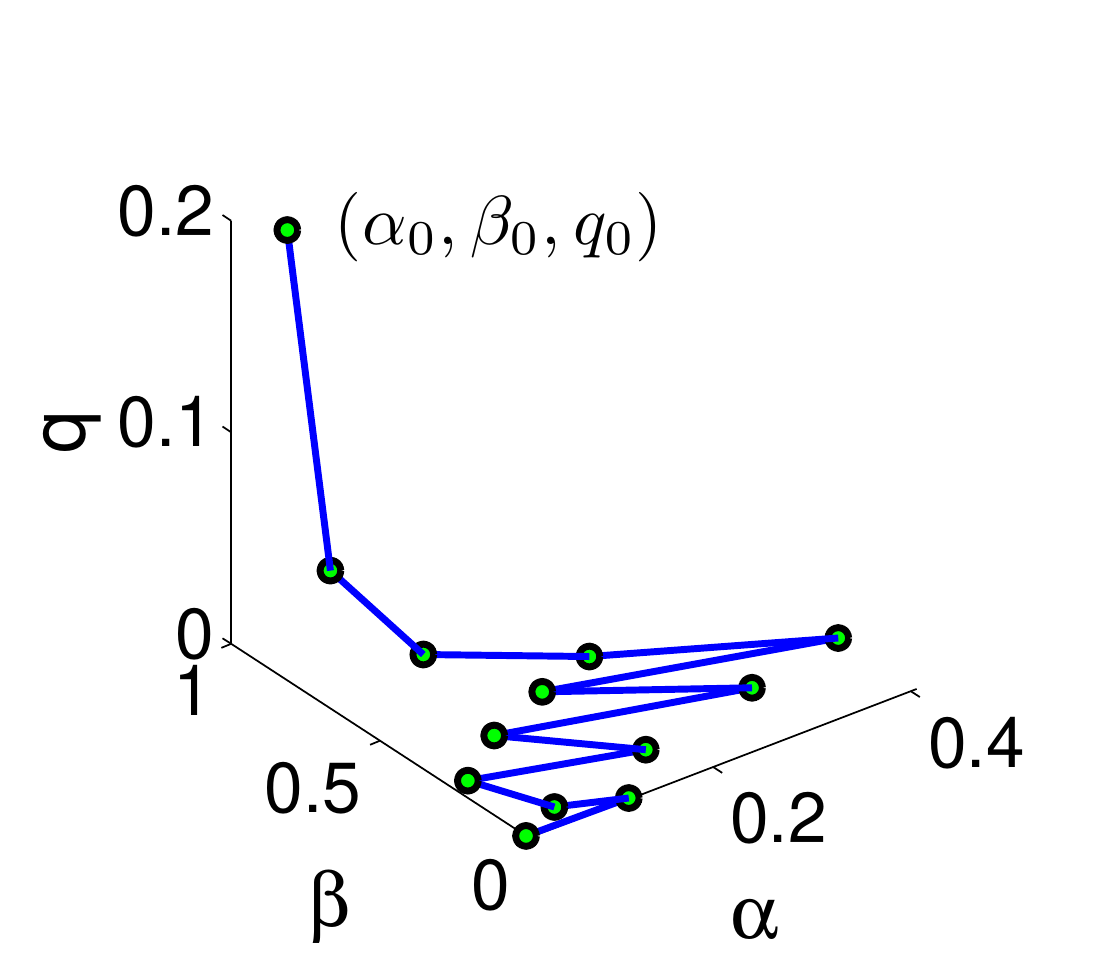}\\ (a)\\
\includegraphics[width=2.2in]{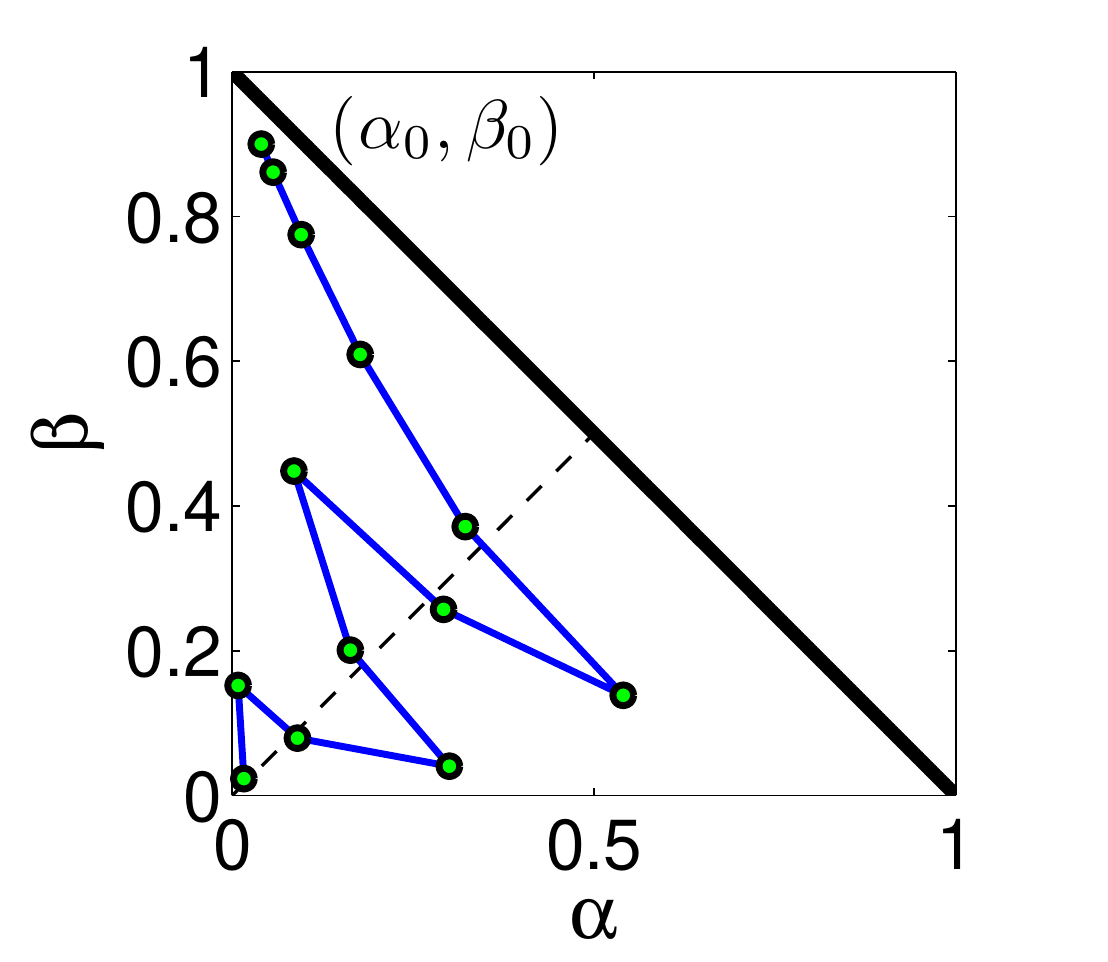}\\
(b)
\end{tabular}
\end{center}
\caption{(a) A typical trajectory of $(\alpha_k,\beta_k,q_k)$ in the $(\alpha, \beta, q)$ coordinates. (b) The trajectory in (a) projected onto the $(\alpha,\beta)$ plane.}
\label{fig:plane}
\end{figure}

\section{The evolution of Type I, Type II, and sensor failure error probabilities}

The relation (\ref{equ:exp}) is symmetric about the hyper-planes $\alpha+\beta = 1$ and $\beta=\alpha$. Thus, it suffices to study the evolution of the dynamic system only in the region bounded by $\alpha+\beta < 1$, $\beta \geq \alpha$, and $0\leq q\leq 1$. Let $\mathcal{U}:=\{(\alpha, \beta)\geq 0|\alpha+\beta<1, \beta\geq\alpha, \text{\space and \space} 0\leq q\leq 1\}$ be this triangular prism. Similarly, define the complementary triangular prism $\mathcal{L}:=\{(\alpha, \beta)\geq 0|\alpha+\beta<1, \beta<\alpha, \text{\space and \space}0\leq q\leq 1 \}$.

First, we denote the following region by $B_1:=\{(\alpha, \beta, q)\in \mathcal{U}|\beta\leq(-q+\sqrt{q^2+(1-q)^2(2\alpha-\alpha^2)+2q(1-q)\alpha})/(1-q)\}$. If $(\alpha_k,\beta_k,q_k)\in B_1$, then the next pair $(\alpha_{k+1},\beta_{k+1},q_{k+1})$  jumps across the plane $\beta=\alpha$ away from $(\alpha_k,\beta_k,q_k)$. More precisely, if $(\alpha_k,\beta_k,q_k)\in \mathcal{U}$, then $(\alpha_k,\beta_k,q_k)\in B_1$ if and only if $(\alpha_{k+1},\beta_{k+1},q_{k+1})\in \mathcal{L}$. This set $B_1$ is identified in Fig.~\ref{fig:b}(a).

It is easy to see from (\ref{equ:exp}) and (\ref{equ:rel}) that, if we start with $(\alpha_0,\beta_0,q_0)\in
\mathcal{U}\setminus B_1$, then before the system enters $B_1$, we have $\alpha_{k+1}>\alpha_k$ and $\beta_{k+1}<\beta_k$. Thus, the system moves towards the $\beta=\alpha$ plane. Therefore, if the sensor number $N$ is sufficiently large, then the system is guaranteed to enter $B_1$.

\begin{figure}[!th]
\begin{center}
\begin{tabular}{cc}
\includegraphics[width=3.5in]{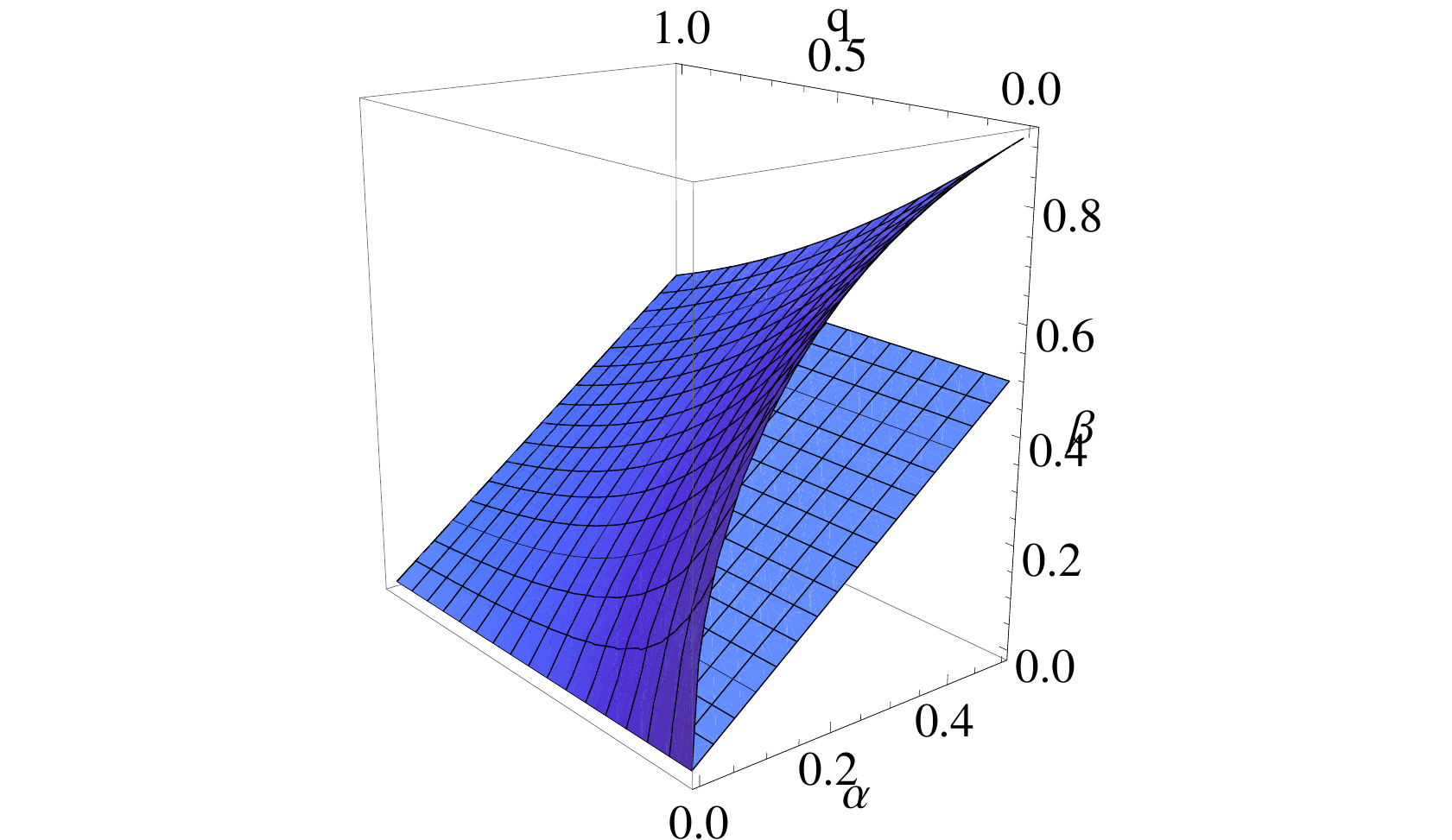}\\
(a)\\
\includegraphics[width=3.5in]{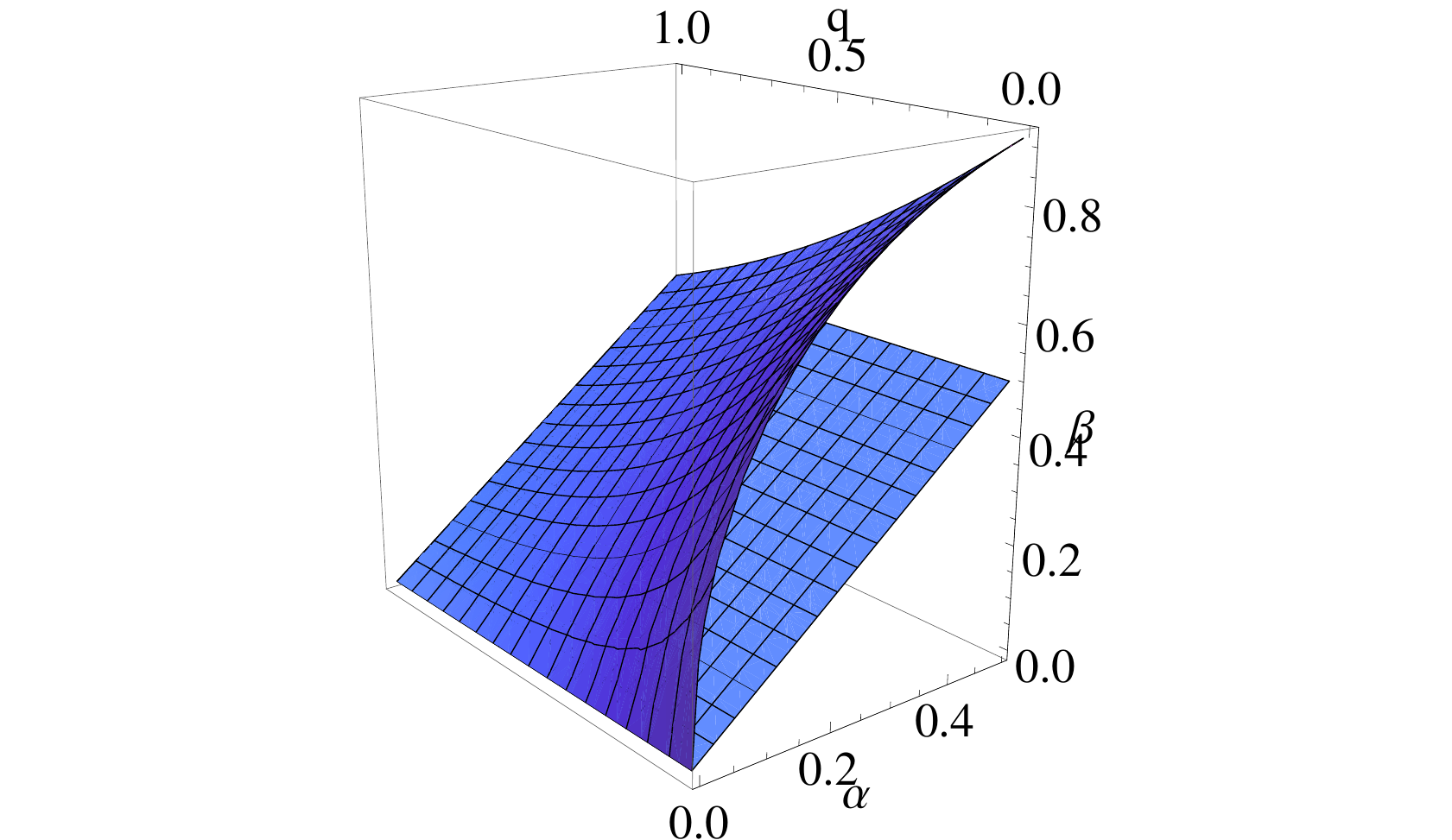}\\
(b)
\end{tabular}
\end{center}
\caption{(a) Region $B_1$ in the $(\alpha, \beta, q)$ coordinates. (b) Region $R_\mathcal{U}$ in the $(\alpha, \beta, q)$ coordinates.}\label{fig:b}
\end{figure}

Next we consider the behavior of the system after it enters $B_1$. If $(\alpha_k, \beta_k,q_k)\in B_1$, we consider the position of the next pair $(\alpha_{k+1}, \beta_{k+1},q_{k+1})$, i.e., consider the \emph{image} of $B_1$ under $f$, denoted by $R_\mathcal{L}$. Similarly we denote the reflection of $R_{\mathcal{L}}$ with respect to $\beta=\alpha$ by $R_{\mathcal{U}}$. We find that $R_\mathcal{U}:=\{(\alpha, \beta, q)\in \mathcal{U}|\beta\leq-\alpha+2(\sqrt{q^2+(1-q^2)\alpha}-q)/(1-q)\}$ (see Fig. \ref{fig:b}(b)).

The sets $R_\mathcal{U}$ and $B_1$ have some interesting properties. We denote the projection of the upper boundary of $R_\mathcal{U}$ and $B_1$ onto the $(\alpha,\beta)$ plane for a fixed $q$ by $R_\mathcal{U}^q$ and $B_1^q$, respectively. It is easy to see that if $q_1\leq q_2$, then $R_\mathcal{U}^{q_1}$ lies above $R_\mathcal{U}^{q_2}$ in the $(\alpha,\beta)$ plane. Similarly, if $q_1\leq q_2$, then $B_1^{q_1}$ lies above $B_1^{q_2}$ in the $(\alpha,\beta)$ plane. Moreover, we have the following Proposition.

\emph{Proposition 1}: $B_1\subset R_\mathcal{U}$.
\begin{proof} $B_1$ and $R_\mathcal{U}$ share the same lower boundary $\beta=\alpha$. Thus, it suffices to proof that the upper boundary of $B_1$ is below that of $R_\mathcal{U}$ for a fixed $q$, i.e., $R_\mathcal{U}^{q}$ lies above $B_1^q$ in the $(\alpha,\beta)$ plane.

The upper boundary of $B_1$ is
\[
\beta=\frac{-q+\sqrt{q^2+(1-q)^2(2\alpha-\alpha^2)+2q(1-q)\alpha}}{1-q}.
\]
The upper boundary of $R_\mathcal{U}$ is
\[
\beta=-\alpha+2\frac{\sqrt{q^2+(1-q^2)\alpha}-q}{1-q}.
\]
Notice that when $q=0$, these boundaries reduce to the boundaries in \cite{Zhang}. We need to prove the following:
\begin{eqnarray*}
&&\frac{-q+\sqrt{q^2+(1-q)^2(2\alpha-\alpha^2)+2q(1-q)\alpha}}{1-q}\\
&\leq&-\alpha+2\frac{\sqrt{q^2+(1-q^2)\alpha}-q}{1-q}.
\end{eqnarray*}
It suffices to show that
\begin{eqnarray*}
&&\sqrt{q^2+(1-q)^2(2\alpha-\alpha^2)+2q(1-q)\alpha}\\
&\leq& -\alpha(1-q)-q+2\sqrt{q^2+(1-q^2)\alpha}.
\end{eqnarray*}
Squaring both sides and simplifying, we have
\begin{eqnarray*}
&&2\sqrt{q^2+(1-q^2)\alpha}(\alpha(1-q)+q) \\
&\leq& 2(q^2+(1-q^2)\alpha)-(1-q)^2(\alpha-\alpha^2).
\end{eqnarray*}
Again squaring both sides and simplifying, we have
\begin{eqnarray*}
&&4(q^2+(1-q^2)\alpha)(q^2+2q(1-q)\alpha+(1-q)^2\alpha^2\\
&&-q^2-(1-q^2)\alpha+(1-q)^2(\alpha-\alpha^2))\\
&\leq& (1-q)^4(\alpha-\alpha^2)^2.
\end{eqnarray*}

Fortuitously, the left hand side turns out to be identically 0. Thus, the inequality holds. The reader can refer to Fig. \ref{fig:RU}(a) and Fig. \ref{fig:RU}(b) for plots of the upper boundaries of $R_\mathcal{U}$ and $B_1$ for two fixed values of $q$.

\end{proof}

\begin{figure}[!th]
\begin{center}
\begin{tabular}{cc}
\includegraphics[width=2.2in]{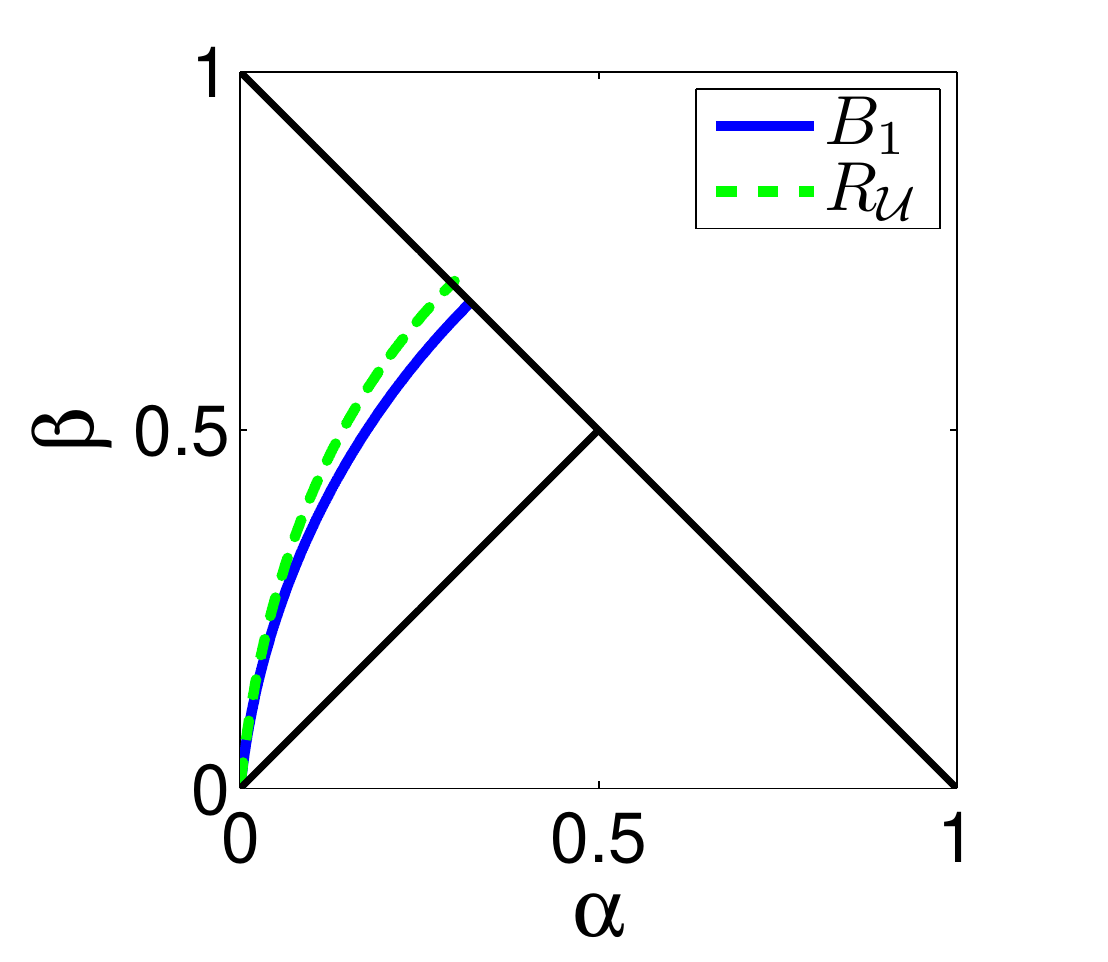} \\
(a)\\
\\\includegraphics[width=2.2in]{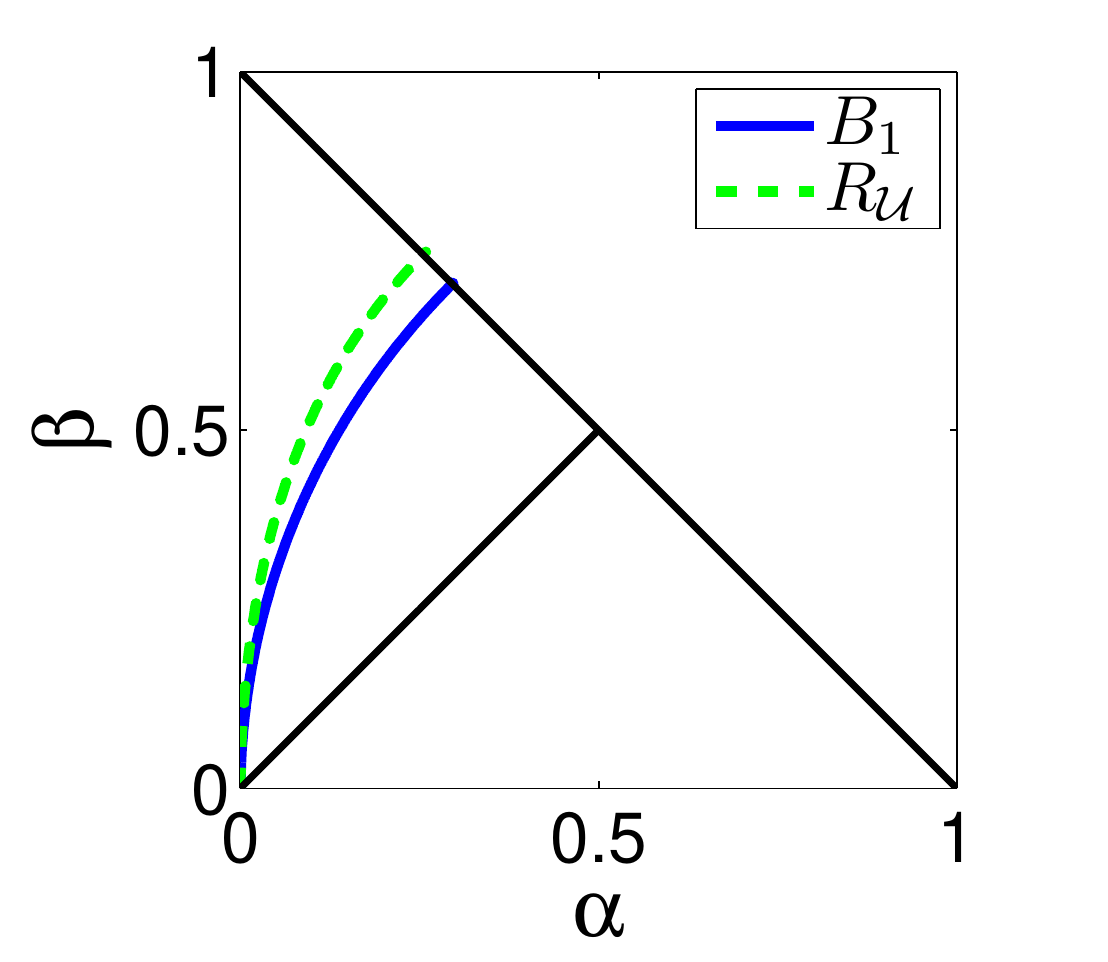}\\
(b)
\end{tabular}
\end{center}
\caption{(a) Upper boundaries for $R_\mathcal{U}$ and $B_1$ for $q=0.1$. (b) Upper boundaries for $R_\mathcal{U}$ and $B_1$ for $q=0.01$.}\label{fig:RU}
\end{figure}
We denote the region $R_\mathcal{U} \cup R_\mathcal{L}$ by $R$. We show below that $R$ is an \emph{invariant region} in the sense that once the system enters $R$, it stays there.

\emph{Proposition 2}: If $(\alpha_{k_0},\beta_{k_0},q_{k_0})\in R$ for some $k_0$, then $(\alpha_k,\beta_k,q_k)\in R$ for all $k\geq k_0$.
\begin{proof} Without lost of generality, we assume $(\alpha_k,\beta_k,q_k)\in R_{\mathcal{U}}$.
We know that $R_{\mathcal{L}}$ is the image of $\mathcal{U}$ in $\mathcal{L}$. Thus if the next state $(\alpha_{k+1},\beta_{k+1},q_{k+1})\in \mathcal{L}$, then it must be inside $R_{\mathcal{L}}$. We already have $q_{k+1}\leq q_k$, which indicates that $R_\mathcal{U}^{q_{k+1}}$ lies above $R_\mathcal{U}^{q_{k}}$ in the $(\alpha,\beta)$ plane. Moreover, for a fixed $q$, the upper boundary $R_\mathcal{U}^q$ is monotone increasing in the $(\alpha,\beta)$ plane. We already know that $\alpha_{k+1}>\alpha_k$ and $\beta_{k+1}<\beta_k$. As a result, if the next state $(\alpha_{k+1},\beta_{k+1},q_{k+1})\in\mathcal{U}$, then the next state is in fact inside $R_\mathcal{U}$.

\end{proof}
We have shown that the system enters $B_1$ after certain levels of fusion. By the fact that $B_1\subset R_\mathcal U$, we conclude that the system enters $R_\mathcal U$ at some level of the tree and stays inside the invariant region $R$ at all levels above.

In the next section, we will consider the step-wise reduction of the total error probability when the system lies inside the invariant region and deduce upper and lower bounds for the total error probability.

\section{Error probability bounds}
The total detection error probability for a node at the $k$th level is $(\alpha_k+\beta_k)/2$ because of the equal-prior assumption. Let $L_k=\alpha_k+\beta_k$, which is twice the total error probability. We will derive bounds on $\log L_k^{-1}$, whose growth rate is related to the rate of converge of $L_k$ to $0$. (Throughout this paper, $\log$ stands for the binary logarithm.)

\emph{Proposition 3}: Let $L_{k+1}^{(q)}$ be the total error probability at the next level from the current state $(\alpha_k,\beta_k,q)$. Suppose that $(\alpha_k,\beta_k,q_1)\text{ and }(\alpha_k,\beta_k,q_2) \in \mathcal{U}$. If $q_1\leq q_2$, then
\[
L_{k+1}^{(q_1)}\leq L_{k+1}^{(q_2)}.
\]
\begin{proof} From (\ref{equ:exp}), we have
\[
L_{k+1}^{(q)}=\frac{1-q}{1+q} L_{k+1}^0+\frac{2q}{1+q} (\alpha_k+\beta_k),
\]
where $L_{k+1}^{(0)}=2\alpha_k-\alpha_k^2+\beta_k^2$.

It is easy to show that $2\alpha_k-\alpha_k^2+\beta_k^2\leq \alpha_k+\beta_k$.
\begin{eqnarray*}
&& 2\alpha_k-\alpha_k^2+\beta_k^2\leq \alpha_k+\beta_k \\
&\Longleftrightarrow& \alpha_k-\alpha_k^2\leq \beta_k-\beta_k^2.
\end{eqnarray*}
Since $\alpha_k+\beta_k\leq 1$ and $\beta_k\geq\alpha_k$, we have $\beta_k-1/2\leq 1/2-\alpha_k$.
Notice that the function $x-x^2$ peaks at $x=1/2$. Hence, $2\alpha_k-\alpha_k^2+\beta_k^2\leq \alpha_k+\beta_k$.

Notice that
\[
\frac{1-q}{1+q}+\frac{2q}{1+q}=1.
\]
Therefore we can write
\[
L_{k+1}^{(q_1)}=p_1 L_{k+1}^0+(1-p_1) (\alpha_k+\beta_k),
\]
where $p_1=(1-q_1)/(1+q_1)$.
Let $p_2=(1-q_2)/(1+q_2)$, it is easy to see that $p_1\geq p_2$. Thus we have
\begin{eqnarray*}
L_{k+1}^{(q_1)}&=&p_1 L_{k+1}^0+(1-p_1) (\alpha_k+\beta_k) \\
&&+(p_2-p_1)L_{k+1}^0-(p_2-p_1)L_{k+1}^0 \\
&\leq& p_1 L_{k+1}^0+(1-p_1) (\alpha_k+\beta_k)\\
&& +(p_2-p_1)L_{k+1}^0-(p_2-p_1)(\alpha_k+\beta_k)\\
&=&L_{k+1}^{(q_2)}.
\end{eqnarray*}
\end{proof}
From Proposition 3, we immediately deduce that
\[
L_{k+1}^{(0)}\leq L_{k+1}^{(q_1)}.
\]
This means that the decay of the total error probability for a single step is the fastest when $q=0$. As a result, for the case where $q\neq 0$, the step-wise shrinkage of the total error probability cannot be faster than the case where $q=0$, where the asymptotic decay exponent is $\sqrt{N}$ \cite{Zhang}.

Notice that from (\ref{equ:q}), the decay of $q_k$ is quadratic, which is much faster than the decay rate of $L_k$. Moreover, it is easy to see that the decay of $q_k$ is faster than the decay of $\alpha_k$ and of $\beta_k$.
Hence, it is natural to assume that $q_k\leq \alpha_k$ and $q_k\leq \beta_k$ when we consider the step-wise shrinkage of the total error probability in the invariant region. Next we give upper and lower bounds for the ratio $L_{k+2}/L_k^2$.

\emph{Proposition 4}: Suppose that $(\alpha_k,\beta_k,q_k)\in R$, $\alpha_k\geq q_k$, and $\beta_k\geq q_k$. Then,
\[
\frac{1}{2}\leq\frac{L_{k+2}}{L_k^2}\leq 4.
\]
\begin{proof} First, we consider the lower bound.
The evolution of the system is
\[
(\alpha_k,\beta_k,q_k)\rightarrow(\alpha_{k+1},\beta_{k+1},q_k^2)\rightarrow(\alpha_{k+2},\beta_{k+2},q_k^4).
\]
From Proposition 3, we have
\[
L_{k+2}^{(0)}\leq L_{k+2},
\]
where $L_{k+2}^{(0)}=2\alpha_{k+1}-\alpha_{k+1}^2+\beta_{k+1}^2$ as defined before.
To prove $1/2\leq L_{k+2}/L_k^2$, it suffices to show that $1/2\leq L_{k+2}^{(0)}/L_k^2$.

If $(\alpha_k,\beta_k)\in R_u\setminus B_1$, then
\[
\frac{L_{k+2}^{(0)}}{L_k^2}=\frac{2\alpha_{k+1}-\alpha_{k+1}^2+\beta_{k+1}^2}{(\alpha_k+\beta_k)^2}.
\]

We have
\[
\alpha_{k+1}=\frac{1-q_k}{1+q_k}(2\alpha_k-\alpha_k^2)+\frac{2q_k}{1+q_k}\alpha_k\geq\alpha_k
\]
and
\[
\beta_{k+1}=\frac{1-q_k}{1+q_k}\beta_k^2+\frac{2q_k}{1+q_k}\beta_k\geq\beta_k^2.
\]

Thus, it suffices to show that
\[
\frac{2\alpha_{k}-\alpha_{k}^2+\beta_{k}^4}{(\alpha_k+\beta_k)^2}\geq \frac{1}{2}.
\]
It is easy to see that
\[
2(2\alpha_{k}-\alpha_{k}^2)\geq 1-(1-\alpha_k)^4.
\]
Hence, it suffices to show that
\[
(1-(1-\alpha_k)^4+\beta_k^4)\geq(\alpha_k+\beta_k)^2,
\]
which has been proved in \cite{Zhang}.

If $(\alpha_k,\beta_k)\in B_1$, then it suffices to show that
\[
\frac{\alpha_{k+1}^2+2\beta_{k+1}-\beta_{k+1}^2}{(\alpha_k+\beta_k)^2}\geq \frac{1}{2}.
\]
We have
\[
\alpha_{k+1}=\frac{1-q_k}{1+q_k}(2\alpha_k-\alpha_k^2)+\frac{2q_k}{1+q_k}\alpha_k\geq\alpha_k
\]
and
\[
\beta_{k+1}=\frac{1-q_k}{1+q_k}\beta_k^2+\frac{2q_k}{1+q_k}\beta_k\geq\beta_k^2.
\]
Thus, it suffices to proof
\[
\frac{\alpha_k^2+\beta_k^2}{(\alpha_k+\beta_k)^2}\geq\frac{1}{2},
\]
which is easy to see.

Next we prove the upper bound of the ratio $L_{k+2}/L_k^2$.


If $(\alpha_k,\beta_k)\in R_u\setminus B_1$, then
\begin{eqnarray*}
&&\frac{L_{k+2}}{L_k^2}\leq\frac{L_{k+1}}{L_k^2}\\
&&=\frac{1-q_k}{1+q_k}\frac{2\alpha_k-\alpha_k^2+\beta_k^2}{(\alpha_k+\beta_k)^2}+\frac{2q_k}{1+q_k}(\alpha_k+\beta_k).
\end{eqnarray*}
It is easy to see that
\[
\frac{2q_k}{1+q_k}(\alpha_k+\beta_k)\leq 1.
\]
Next, we can prove that
\[
\frac{2\alpha_k-\alpha_k^2+\beta_k^2}{(\alpha_k+\beta_k)^2}\leq 2,
\]
which is equivalent to
\[
\phi(\alpha_k,\beta_k):=2\alpha_k-3\alpha_k^2-\beta_k^2-4\alpha_k\beta_k\leq 0.
\]
We have
\[
\frac{\partial \phi}{\partial \beta_k}=-2\beta_k-4\alpha_k\leq0.
\]
Thus, we can consider the lower boundary of this region which is the upper boundary of $B_1$.
\[
\beta=\frac{-q+\sqrt{q^2+(1-q)^2(2\alpha-\alpha^2)+2q(1-q)\alpha}}{1-q}.
\]
Denote $\varphi(\alpha,q):=\sqrt{q^2+(1-q)^2(2\alpha-\alpha^2)+2q(1-q)\alpha}$. We have
\begin{eqnarray*}
\phi(\alpha_k,\beta_k)&=&-(q_k^2+q_k^2+(1-q_k)^2(2\alpha_k-\alpha_k^2)\\
&&+2q_k(1-q_k)\alpha_k-2q_k\varphi(\alpha_k,q_k))/(1-q_k)^2\\
&&-4\alpha_k\beta_k+2\alpha_k-3\alpha_k^2\\
&=& \frac{2q_k\beta_k}{1-q_k}-4\alpha_k\beta_k-\frac{2q_k\alpha_k}{1-q_k}-2\alpha_k^2 .
\end{eqnarray*}
It is easy to see that
\[\frac{2q_k\beta_k}{1-q_k}-4\alpha_k\beta_k\leq 0.\]
Hence, we have
\[
\frac{1-q_k}{1+q_k}\frac{2\alpha_k-\alpha_k^2+\beta_k^2}{(\alpha_k+\beta_k)^2}\leq 2,
\]
and
\[\frac{L_{k+2}}{L_k^2}\leq 3.\]

For the case where $(\alpha_k, \beta_k)\in B_1$, we prove that the ratio is upper bounded by $4$. The evolution of the system is
\[
(\alpha_k,\beta_k,q_k)\rightarrow(\alpha_{k+1},\beta_{k+1},q_k^2)\rightarrow(\alpha_{k+2},\beta_{k+2},q_k^4).
\]
It is easy to see that
\[
L_{k+2}^{(q_k)}\geq L_{k+2},
\]
where $L_{k+2}^{(q_k)}$ denotes the total error probability if we use $q_k$ to calculate from $L_{k+1}$ to $L_{k+2}$. Therefore, it suffices to prove that
\[
L_{k+2}^{(q_k)}-4L_{k}^2=\alpha_{k+2}+\beta_{k+2}-4(\alpha_k+\beta_k)^2\leq 0.
\]
We have
\[
\beta_{k+1}=\frac{1-q_k}{1+q_k}\beta_k^2+\frac{2q_k}{1+q_k}\beta_k.
\]
From the assumption that $\beta_k\geq q$, we have
\[
\frac{\partial \beta_{k+1}}{\partial \beta_k}=\frac{2(1-q_k)}{1+q_k}\beta_k+\frac{2q_k}{1+q_k}\leq 4\beta_k.
\]
It is easy to get that
\begin{eqnarray*}
\beta_{k+2}&=&\frac{1-q_k}{1+q_k}(2\beta_{k+1}-\beta_{k+1}^2)+\frac{2q_k}{1+q_k}\beta_{k+1}\\
&=&-\frac{1-q_k}{1+q_k}\beta_{k+1}^2+\frac{2}{1+q_k}\beta_{k+1}.
\end{eqnarray*}
Therefore, we have
\[
\frac{\partial \beta_{k+2}}{\partial \beta_k}=-2\frac{1-q_k}{1+q_k}\beta_{k+1}\frac{\partial \beta_{k+1}}{\partial \beta_k}+\frac{2}{1+q_k}\frac{\partial \beta_{k+1}}{\partial \beta_k}\leq 8\beta_k.
\]
Thus,
\[
\frac{\partial L_{k+2}^{(q_k)}-4L_{k}^2}{\partial \beta_k}\leq 8\beta_k-8\alpha_k-8\beta_k\leq 0.
\]
Therefore, we can consider the lower boundary of $B_1$, $\beta_k=\alpha_k$. We have
%
%
\begin{eqnarray*}
L_{k+2}^{(q_k)}-4L_{k}^2=\frac{4(1-q_k)^2(1-q_k)}{(1+q_k)^3}\alpha_k^2
-4\frac{(1-q_k)^2}{(1+q_k)^2}\alpha_k^3\\+\frac{2(1-q_k)^2}{(1+q_k)^2}\alpha_k^2+\frac{8q_k}{(1+q_k)^2}\alpha_k-16\alpha_k^2\leq 0,
\end{eqnarray*}
which holds in region $B_1$. Hence, the ratio is upper bounded by $4$ in this region.

\end{proof}

Proposition 4 gives rise to bounds on the change in the total error probability every two steps: $L_{k+2}\leq  4L_{k}^2$ and $L_{k+2}\geq L_k^2/2$. From these, we can derive bounds for $\log L_k^{-1}$ for even-height trees, i.e., $k=\log N$ is even. Let $P_N=L_{\log N}$, namely, the total error probability at the fusion center. We will derive bounds for $\log P_N^{-1}$.
\begin{Theorem}
If $(\alpha_0,\beta_0,q_0)\in R$ and $\log N$ is even, then
\begin{eqnarray*}
\sqrt{N} \left(\log L_0^{-1} - \frac{2\log{\sqrt{N}}}{\sqrt{N}}\right)
\leq \log P_N^{-1}\\ \leq \sqrt{N}\left(\log L_0^{-1}+\frac{\lsN}{\sqrt N}\right).
\end{eqnarray*}
\end{Theorem}

\begin{proof} If $(\alpha_0, \beta_0,q_0) \in R$, then we have $(\alpha_k, \beta_k,q_k)\in R$ for $k=0,1,\ldots,\log N-2$. From Proposition 4, we have
\[
L_{k+2}=a_k  L_{k}^2
\]
for $k=0,1,\ldots,\log N-2$ and some $a_k \in [1/2,4]$. Therefore, for $k=2,4,\ldots, \log N $, we have
\[
L_k = \left(\prod_{i=1}^{k/2} a_i\right) L_0^{2^{k/2}},
\]
where $a_i \in [1/2,4]$. Substituting $k=\log N$, we have
\[
P_N
= \left(\prod_{i=1}^{\lsN} a_i\right) L_0^{2^{\lsN}}
= \left(\prod_{i=1}^{\lsN} a_i\right) L_0^{\sqrt{N}}.
\]
Hence,
\begin{align*}
\log P_N^{-1}
&= -\left(\sum_{i=1}^{\lsN} \log a_i\right) + \sqrt{N}\log L_0^{-1}.
\end{align*}
Notice that $\log L_0^{-1}>0$ and for each $i$,
$-1 \leq \log a_i \leq 2$. Thus,
\begin{align*}
\log P_N^{-1} &\leq \sqrt{N}\log L_0^{-1}+\lsN\\
&=\sqrt{N}\left(\log L_0^{-1}+\frac{\lsN}{\sqrt N}\right).
\end{align*}
Finally,
\begin{align*}
\log P_N^{-1}
&\geq  -2\lsN+ \sqrt{N}\log L_0^{-1}\\
&= \sqrt{N} \left(\log L_0^{-1} -
\frac{2\lsN}{\sqrt{N}}\right).
\end{align*}
\end{proof}

For odd-height trees, we need to calculate the decrease in the total error probability in a single step. For this, we have the following Proposition.

\emph{Proposition 5}: If $(\alpha_k,\beta_k,q_k)\in \mathcal{U}$, then we have
\[
\frac{L_{k+1}}{L_k^2}\geq 1
\] and
\[
\frac{L_{k+1}}{L_k}\leq 1.
\]
\begin{proof}
To prove $L_{k+1}/L_k^2\geq1$, it suffices to prove that
\begin{align*}
&\frac{1-q_k}{1+q_k}(2\alpha_k-\alpha_k^2+\beta_k^2-(\alpha_k+\beta_k)^2)\\
&+\frac{2q_k}{1+q_k}(\alpha_k+\beta_k-(\alpha_k+\beta_k)^2) \geq 0,
\end{align*}
which is easy to see.

To prove $L_{k+1}/L_k\leq1$, it suffices to prove that
\begin{align*}
&\frac{1-q_k}{1+q_k}(2\alpha_k-\alpha_k^2+\beta_k^2-(\alpha_k+\beta_k))\\
&+\frac{2q_k}{1+q_k}(\alpha_k+\beta_k-(\alpha_k+\beta_k)) \leq 0,
\end{align*}
which is easy to see.

\end{proof}

From Propositions 4 and 5, we give bounds for the total error probability at the fusion center for trees with odd height.

\begin{Theorem}If $(\alpha_0,\beta_0,q_0)\in R$, then
\begin{eqnarray*}
\sqrt{\frac{N}{2}} \left(\log L_0^{-1} -  \frac{2\log{\sqrt{\frac{N}{2}}}}{\sqrt{\frac{N}{2}}}\right)
\leq -\log P_N \\
\leq \sqrt{2N}\left(\log L_0^{-1}+\frac{\log{\sqrt{\frac{N}{2}}}}{\sqrt {2N}}\right).
\end{eqnarray*}
\end{Theorem}
\begin{proof}
By Proposition 5, we have
\[
L_{1}=\widetilde{a} L_{0}^2
\]
for some $\widetilde{a} \geq 1$.

By Proposition 4, we have
\[
L_{k+2}=a_k L_{k}^2
\]
for $k=1,3,\ldots,\log N-2$ and some $a_k \in [1/2,4]$. Hence, we can write
\[
L_{k}=\widetilde{a}\left(\prod_{i=1}^{(k-1)/2} a_i\right) L_0^{2^{(k+1)/2}},
\]
where $1/2\leq a_i\leq 4$ for $i=1,2,\ldots,(k-1)/2$ and $\widetilde{a}\geq 1$.
Let $k=\log N$, we have
\[
P_N
= \widetilde{a}\left(\prod_{i=1}^{\log{\sqrt{\frac{N}{2}}}} a_i\right) L_0^{2^{\log{\sqrt{2N}}}}
= \widetilde{a}\left(\prod_{i=1}^{\log{\sqrt{\frac{N}{2}}}} a_i\right) L_0^{\sqrt{2N}},
\]
and so
\begin{align*}
\log P_N^{-1}
&= -\log \widetilde{a}-\left(\sum_{i=1}^{\log{\sqrt{\frac{N}{2}}}} \log a_i\right) + \sqrt{2N}\log L_0^{-1}.
\end{align*}
Notice that $\log L_0^{-1}>0$ and for each $i$,
$ \log a_i \geq -1$. Moreover, $ \log\widetilde{a}\geq0 $. Hence,
\begin{eqnarray*}
\log P_N^{-1} &\leq& \sqrt{2N}\log L_0^{-1}+\log{\sqrt{\frac{N}{2}}}\\
&=&\sqrt{2N}\left(\log L_0^{-1}+\frac{\log{\sqrt{\frac{N}{2}}}}{\sqrt {2N}}\right).
\end{eqnarray*}

By Proposition 5, we can write
\[
L_{1}=\widetilde{a} L_{0}
\]
for  some $\widetilde{a} \leq 1$. Thus,
\[
L_{k}=\widetilde{a}\left(\prod_{i=1}^{(k-1)/2} a_i\right) L_0^{2^{(k-1)/2}},
\]
where $1/2\leq a_i\leq 4$ for $i=1,2\ldots,(k-1)/2$ and $ \widetilde{a}\leq1$.
Hence,
\[
P_N
= \widetilde{a}\left(\prod_{i=1}^{\log{\sqrt{\frac{N}{2}}}} a_i\right) L_0^{2^{\log{\sqrt{\frac{N}{2}}}}}
= \widetilde{a}\left(\prod_{i=1}^{\log{\sqrt{\frac{N}{2}}}} a_i\right) L_0^{\sqrt{\frac{N}{2}}}
\]
and so
\begin{align*}
\log P_N^{-1}
&= -\log \widetilde{a}-\left(\sum_{i=1}^{\log{\sqrt{\frac{N}{2}}}} \log a_i\right) + \sqrt{\frac{N}{2}}\log L_0^{-1}.
\end{align*}

Notice that $\log L_0^{-1}>0$ and for each $i$, $-1 \leq \log a_i \leq 2 $ and $\log {\widetilde{a}} \leq 0 $. Thus,
\begin{eqnarray*}
\log P_N^{-1}
&\geq& -2\log{\sqrt{\frac{N}{2}}} + \sqrt{\frac{N}{2}}\log L_0^{-1} \\
&=& \sqrt{\frac{N}{2}} \left(\log L_0^{-1} -
\frac{2\log{\sqrt{\frac{N}{2}}}}{\sqrt{\frac{N}{2}}}\right).
\end{eqnarray*}

\end{proof}

\section{Asymptotic Rates}
In this section, we first consider the asymptotic decay rate of the total error probability with respect to $N$. We compare the rate with that of balanced binary relay trees without sensor failures. Then we allow the sensors to be asymptomatically bad, in the sense that $q_0\rightarrow 1$ and $\alpha_0+\beta_0\rightarrow 1$. We prove that the total error probability still converges to $0$ provided the convergence of $q_0$ and $\alpha_0+\beta_0$ is sufficiently slow.
\subsection{Asymptotic decay rate}

 Notice that when $N$ is very large,
the sequence $\{(\alpha_k, \beta_k,q_k)\}$ enters the invariant region $R$ at some level and stays inside afterward. Therefore the decay rate in the invariant region determines the asymptotic rate. Because our error probability bounds for odd-height trees differ from those of even-height trees by a constant term,
without lost of generality, we will consider trees with even height to calculate the decay rate.

\emph{Proposition 6}: If $L_0=\alpha_0+\beta_0$ is fixed, then
\begin{equation*}
\log P_N^{-1} \sim \log L_0^{-1}\sqrt{N}.
\end{equation*}

\begin{proof}
If $L_0=\alpha_0+\beta_0$ is fixed, then by Theorem 1 we immediately see that $P_N\to 0$ as $N\to\infty$ ($\log P_N^{-1}\to\infty$) and
\begin{eqnarray*}
1-2\frac{\log{\sqrt N}}{\log L_0^{-1}\sqrt N} \leq\frac{\log P_N^{-1}}{\log L_0^{-1}\sqrt{N}}
 \leq 1+\frac{\log{\sqrt N}}{\log L_0^{-1}\sqrt N}.
\end{eqnarray*}
In addition, because $\log{\sqrt{N}}/\sqrt{N}\to 0$, we have
\[
\frac{\log P_N^{-1}}{\log L_0^{-1}\sqrt{N}}\rightarrow 1,
\]
which means
\begin{equation*}
\log P_N^{-1} \sim \log L_0^{-1}\sqrt{N}.
\end{equation*}
\end{proof}
This implies that the convergence of the total error probability is sub-exponential with decay exponent $\sqrt N$. Compared to the decay exponent for the case where $q=0$ (no sensor failures), the asymptotic rate does not change when we have crummy sensors, even though the step-wise shrinkage for the crummy sensor case is worse.

Given $L_0\in(0,1)$ and $\ep\in (0,1)$, suppose that we wish to determine
how many sensors we need to have so that $P_N\leq\ep$. The solution is simply to find an
$N$ (e.g., the smallest) satisfying the inequality
\[
\sqrt{N} \left(\log L_0^{-1} - \frac{2\log{\sqrt{N}}}{\sqrt{N}}\right)
\geq -\log\ep.
\]
The smallest $N$ grows like $\Theta((\log\ep)^2)$ (cf., \cite{Zhang}, in which the growth rate is the same, and \cite{Gubner}, where a looser bound was derived).

\subsection{Asymptotically bad sensors}
First we consider the case where $q_0$ depends on $N$ (denoted by $q_0^{(N)}$). We wish to have the failure error probability at the fusion center $q_{\log N}$ to converge to $0$.

If $q_0^{(N)}$ is bounded by some constant $q \in (0,1)$ for all $N$, then clearly $q_{\log N}\to 0$. So henceforth suppose that $q_0^{(N)}\rightarrow 1$, which means that the sensors are asymptotically arbitrarily unreliable.

\emph{Proposition 7}:
Suppose that  $q_0^{(N)}=1-\eta_N$ with $\eta_N\to 0$.
Then, $q_{\log N}\rightarrow 0$ if and only if $\eta_N=\omega(1/N)$ (i.e., $\eta_N N\rightarrow \infty$).

\begin{proof}
From (\ref{equ:q}), we have
\[
q_k=(q_0^{(N)})^{2^k}.
\]
Letting $k=\log N$, we can write
\[
q_{\log N}=(q_0^{(N)})^N,
\]
or equivalently,
\[
\log q_{\log N} ^{-1}= N\log \left((q_0^{(N)})^{-1}\right).
\]
It is easy to see that $q_{\log N}\rightarrow 0$ if and only if $N\log(1-\eta_N)^{-1} \rightarrow \infty$.
But as $x\to 0$, $-\log(1-x)\sim x/\ln(2)$. Hence,
$q_{\log N}\to 0$ if and only if $\eta_N N\to \infty$.

\end{proof}
Now suppose that $c_1/N\leq \eta_N\leq c_2/N$. In
this case, for large $N$ we deduce that
\[
c_1\leq \log q_{\log N} ^{-1}\leq c_2,
\]
or equivalently,
\[
2^{-c_2}\leq q_{\log N}\leq 2^{-c_1}.
\]

Finally, if $\eta_N=o(1/N)$ (i.e., $\eta_N$ converges to $0$
strictly faster than $1/N$), then $q_{\log N}\to 1$.

Next we allow the detection error probability of individual sensors to depend on $N$, denoted by $L_0^{(N)}$.

If $L_0^{(N)}$ is bounded by some constant $L \in (0,1)$ for all $N$, then clearly $P_N\to 0$.
It is more interesting to consider $L_0^{(N)}\to 1$, which means that sensors are asymptotically bad.

\emph{Proposition 8}:
Suppose that $L_0^{(N)}=1-\eta_N$ with
$\eta_N\to 0$.
Then, $P_N\rightarrow 0$ if and only if $\eta_N=\omega(1/\sqrt{N})$.

\begin{proof}
For
sufficiently large $N$,
\[
\sqrt{N}  \frac{\log \left((L_0^{(N)})^{-1}\right)}{2}
\leq \log P_N^{-1} \leq 2\sqrt{N}\log \left((L_0^{(N)})^{-1}\right).
\]
We conclude that $P_N\to 0$ if and only if
\[
\sqrt{N}\log\left( (L_0^{(N)})^{-1}\right)\to \infty.
\]
Therefore,
\[
\sqrt{N}\log \left((L_0^{(N)})^{-1}\right) = -\sqrt{N}\log (1-\eta_N).
\]
But as $x\to 0$, $-\log(1-x)\sim x/\ln(2)$. Hence,
$P_N\to 0$ if and only if $\eta_N\sqrt{N}\to \infty$ or $\eta_N=\omega(1/\sqrt{N})$.

\end{proof}
Now suppose that $c_1/\sqrt{N}\leq \eta_N\leq c_2/\sqrt{N}$. In
this case, for large $N$ we deduce that
\[
c_1\leq \log P_N^{-1}\leq c_2,
\]
or equivalently,
\[
2^{-c_2}\leq P_N\leq 2^{-c_1}.
\]

Finally, if $\eta_N=o(1/\sqrt{N})$ (i.e., $\eta_N$ converges to $0$
strictly faster than $1/\sqrt{N}$), then $P_N\to 1$.

\section{Conclusion}
We have studied the detection performance of balanced binary relay trees with crummy sensors. We have shown that there exists an invariant region in the space of $(\alpha, \beta, q)$ triplets. We have also developed total error probability bounds at the fusion center as  functions of $N$ for both even-height trees and odd-height trees.
These bounds imply that the total error probability converges to $0$ sub-linearly, with a decay exponent that is essentially $\sqrt N$. Compared to balanced binary relay trees with no sensor failures, the step-wise shrinkage of the total error probability for the crummy sensor case is slower, but the asymptotic decay rate is the same.
In addition, we allow all sensors to be asymptotically bad, in which case we deduce necessary and sufficient conditions for the total error probability to converge to $0$.

\bibliographystyle{IEEEbib}

\end{document}